\documentclass{CSML}
\pdfoutput=1

\usepackage{lastpage}
\newcommand{\dOi}{13(4:6)2017}
\lmcsheading{}{1--\pageref{LastPage}}{}{}%
{Sep.~08,~2016}{Nov.~06, 2017}{}


\usepackage{hyperref}
\hypersetup{hidelinks}
\theoremstyle{plain} 

\def\bbN{{\mathbb N}}

\def\ra{\rightarrow}
\def\nats{{\bbN}}
\def\voc{\tau}

\newcommand{\Am}{\stra}

\newcommand{\Aa}{\Am}

\newcommand{\Tt}{\ensuremath{\mathbb{T}}}

\newcommand{\la}{\leftarrow}

\def\spec{{\rm{spec}}}
\def\calC{{\mathcal C}}
\def\calR{{\mathcal R}}
\def\stra{{\mathfrak A}}

\def\lma{T_\leftarrow}
\def\rma{T_\rightarrow}
\def\lmt{T_\leftarrow^+}
\def\rmt{T_\rightarrow^+}

\def\lfa{f_\leftarrow}
\def\rfa{f_\rightarrow}
\def\Lfa{F_\leftarrow}
\def\Rfa{F_\rightarrow}
\def\ufa{f_\uparrow}
\def\uan{f_\downarrow}

\def\binntime{{\bf NTIME}_2}
\def\binnts{{\bf NTS}_2}
\def\binntisp{{\bf NTISP}_2}
\def\PIFSpec{{\bf PIFSpec}}
\def\BDSpec{{\bf BDSpec}}
\def\QTL{{\bf QTL}}
\def\PPifSpec{{\bf PPifSpec}}
\def\FPPifSpec{{\bf FPPifSpec}}
\def\PBDSpec{{\bf PBDSpec}}
\def\PSpec{{\bf PSpec}}
\def\NP{{\bf NP}}
\def\NE{{\bf NE}}
\def\P{{\bf P}}
\def\NEXPTIME{{\bf NEXPTIME}}

\begin{document}

\title{Bounded degree and planar spectra}

\author[A.~Dawar]{Anuj Dawar}	
\address{Department of Computer Science and Technology, University of Cambridge, J.J. Thomson Avenue, Cambridge CB3 0FD, England}	
\email{anuj.dawar@cl.cam.ac.uk}  

\author[E.~Kopczy\'nski]{Eryk Kopczy\'nski}	
\address{Institute of Informatics, University of Warsaw, Banacha 2, 02-097 Warszawa, Poland}	
\email{erykk@mimuw.edu.pl}  





\begin{abstract}
  \noindent 
The finite spectrum of a first-order sentence is the set of positive 
integers that are the sizes of its models.  The class of finite spectra 
is known to be the same as the complexity class NE.  We consider 
the spectra obtained by limiting models to be either planar (in the
graph-theoretic sense) or by bounding the degree of elements.  We show
that the class of such spectra is still surprisingly rich by establishing
that significant fragments of NE are included among them.  At the
same time, we establish non-trivial upper bounds showing that not all sets
in NE are obtained as planar or bounded-degree spectra.
\end{abstract}

\maketitle

\section{Introduction}

The \emph{spectrum} of a sentence $\phi$ of some logic, denoted
$\spec(\phi)$, is the set of positive integers $n$ such that $\phi$
has a model of cardinality $n$.  In this paper we are solely concerned
with first-order logic and we use the word spectrum to mean a set of
integers that is is the spectrum of some first-order sentence.  Scholz
in~\cite{scholz} posed the question of characterizing those sets of
integers which are spectra.  This question has spawned a large amount
of research.  In particular, Fagin's attempt to answer Asser's
question of whether the complement of a spectrum is itself a spectrum
launched the field of descriptive complexity theory.  An excellent
summary of the history of the spectrum question and methods used to
approach it are given in the survey by Durand et
al.~\cite{fiftyyears}.

An exact characterization of the sets of spectra of first-order
sentences can be given in terms of complexity theory.  This result,
obtained by Fagin~\cite{fagin74} and by Jones and
Selman~\cite{JonesS74} states that a set $S \subseteq \nats$ is a
spectrum if, and only if, the binary representations of integers in
$S$ can be recognised by a non-deterministic Turing machine running in
time $2^{O(n)}$.  In the language of complexity theory, we say that
the class of spectra is exactly the complexity class $\NE$.

Besides characterizations of the class of spectra, researchers have
also investigated how the class is changed by restrictions, either on
the form of the sentence $\phi$ or on the class of finite models that
we consider.  A particularly fruitful line of research in the former
direction has investigated the spectra of sentences in a restricted
vocabulary.  It is an easy observation that if the vocabulary includes
only unary relations, then the spectrum of any sentence is either
finite or co-finite.  Durand et al.~\cite{onefunction} investigate the
spectra of sentences in vocabularies with unary relations and one
unary function and show that the class of such spectra is exactly the
ultimately periodic (or semilinear) sets of natural numbers.  On the
other hand, if we allow two unary functions or one binary relation in
the vocabulary, there are spectra that are \NEXPTIME-complete
\cite{fagin74, onebinary}.  However, since the class of these spectra
is not closed under polynomial-time reductions, it does not follow
that it includes all of $\NEXPTIME$ or indeed even all of $\NE$.
Whether every spectrum spectrum $S$ is the spectrum of a sentence in
such a vocabulary remains an open question.  The results
of~\cite{onefunction} have also been extended beyond first-order logic
to monadic second-order logic by Gurevich and Shelah~\cite{gs03}.

In the present paper, we investigate the spectra that can be obtained
by restricting the class of models considered.  To be precise, we
consider restrictions on the Gaifman (or adjacency) graph of the
structure.  We want to explore what reasonable restrictions on such
graphs make for a more tractable class of spectra.  This is in the
spirit of a recent trend in finite model theory aimed at investigating
it on \emph{tame} classes of structures (see~\cite{Daw07}). 

Fischer and Makowsky~\cite{mak} show that if  $K$ is a class of
structures generated by a certain class of graph grammars known as
$eNCE$-grammars (see~\cite{Kim97}) and all models of $\phi$
are in $K$, then $\spec(\phi)$ is ultimately periodic.  Here, $\phi$
can be a formula of monadic second order logic with counting
(CMSO).  Classes of graphs generated by $eNCE$-grammars include $TW_d$,
the class of structures of tree-width at most $d$, and $CW_d$, the
structures of clique width at most $d$.  It is not difficult to show
that every semilinear set of natural numbers can be obtained as such a
spectrum.  

The restricted classes of structures we investigate in this paper are
those of bounded degree and those whose Gaifman graph is planar.  We
also consider the two restrictions in combination.  There are two ways
that we can define spectra of first-order sentences restricted to a
class of structures $\calC$.  We can define the class of
$\calC$-spectra as the spectra of those sentences $\phi$ all of whose
models are in $\calC$.  Alternatively, we can define the class of
$\calC$-spectra as the class of sets $S\subseteq \nats$ such that
there is a first-order sentence $\phi$ and $n \in S$ iff there is a
model of $\phi$ in $\calC$ of cardinality $n$.  If $\calC$ is itself
first-order definable (for instance if it is the class of graphs of
degree at most $k$ for some fixed integer $k$), then the two notions
coincide.  Otherwise, the first is more restrictive than the second.

Our lower bound results show that the class of spectra obtained by
restricting to bounded degree and to planar graphs are still quite
rich.  In particular, Theorem~\ref{mainlower} shows that the class of
spectra of structures of degree at most 3 still contains all sets $S$
of integers for which $N\in S$ is decidable by a nondeterministic
machine running in time $O(N)$ (which is time exponential in the
binary representation of $N$).  Similar (though somewhat weaker) lower
bounds are established for classes of structures which are required to
be both planar and of bounded degree in Theorems~\ref{clocksim}
and~\ref{queuemach}.  

On the other hand, the upper bounds show that the class of spectra we
obtain is strictly weaker than the class of all spectra.  In
particular we establish, for the bounded degree restriction, in 
Theorem~\ref{mainupper} a complexity upper bound on the class of
spectra obtained.  This,  along with the non-deterministic time hierarchy
theorem~\cite{SFM78} establishes that this class of spectra is not the
whole class $\NE$.
Theorem~\ref{planarupper} establishes a related 
complexity upper bounds on the class of spectra for bounded degree
planar graphs.  Dropping the degree restriction, we can obtain a
complexity upper bound on planar spectra from results of Frick and
Grohe~\cite{FG01} from which again it follows that this is not the
class of all spectra, answering
negatively Open Question~6 from~\cite{fiftyyears}.

\section{Preliminaries}

For a relational vocabulary $\tau$ (i.e.\ a finite collection of
constant and relation symbols), a $\tau$-structure $\Aa$ is a set $A$
along with an interpretation in $A$ of every symbol in $\tau$.  We
will confine ourselves to vocabularies $\tau$ in which all symbols are
binary or unary.  The Gaifman graph $\Gamma(\Aa)$ of $\Aa$ is the
graph on the vertex set $A$ such that two distinct vertices $u$ and
$v$ are adjacent if, and only if, either $(u,v)$ or $(v,u)$ occurs in
the interpretation in $\Aa$ of some binary relation in $\tau$.  We say
that the degree of $\Aa$ is bounded by $d$ if, in $\Gamma(\Aa)$ no
vertex is adjacent to more than $d$ others.  Similarly, we say that
$\Aa$ is planar if $\Gamma(\Aa)$ is.

\subsection{Partial function symbols}
For the most part, we will consider structures in which every binary
relation in $\tau$ is interpreted as the graph of a \emph{partial
  injective unary function} (PIF).  To emphasise this, we use
functional notation and write, for instance, $t_1(x) = t_2(y)$ as
shorthand for the formula that says that the terms $t_1(x)$ and
$t_2(y)$ are defined and their interpretations are equal.  Thus, for
example, $f(x) = f(x)$ is equivalent to saying that $f(x)$ is defined.

A structure interpreting $d$ such partial injective function symbols
is always of degree at most $2d$.  Conversely, if we have a graph of
degree $d$, by Vizing's theorem~\cite{Vizing} we can colour the edges with $d+1$ colours in a way such
that all edges incident with a given vertex are of different colours.
We can then represent such structures using $2d-1$ partial injective
function symbols.  This establishes that in considering the spectra of
first-order sentences over bounded-degree structures, there is no loss
of generality in considering structures in which all binary relations
are PIFs.

\begin{defi}\label{PIFSpecd}
Let $\PIFSpec_d$ denote the set of $S \subseteq \bbN$ such that $S$ is
a spectrum of some sentence $\phi$ using only unary relations and $d$
PIF symbols.
\end{defi}

\begin{prop}\label{propclosure}
If $S \in \PIFSpec_d$, then the following are also in $\PIFSpec_d$:
\begin{enumerate}
\item $S \cup \{n\}$, \label{en:first}
\item $S - \{n\}$, \label{en:second}
\item $\{x+1: x \in S\}$, \label{en:third}
\item $\{x \in \bbN: x+1 \in S\}$. \label{en:fourth}
\end{enumerate}
\end{prop}
\begin{proof}
Let $\phi$ be a sentence whose spectrum is $S$, witnessing that $S \in
\PIFSpec_d$.  Also let $\eta_n$ be the first-order sentence that
states that there are exactly $n$ distinct elements in the universe.
Then,~(\ref{en:first}) is obtained as the spectrum of the sentence
$\phi \vee \eta_n$ and~(\ref{en:second}) is obtained as the spectrum
of $\phi \land \neg\eta_n$.

To establish~(\ref{en:third}), extend the vocabulary by a new unary
relation symbol $R$ and let $\rho$ be the sentence that says that
there is exactly one element not in $R$.  Now, consider the spectrum
of the sentence $\rho \land \phi^R$, where $\phi^R$ denotes the
sentence obtained by relativising all quantifiers in $\phi$ to $R$.

Finally, for a structure $\Aa$ and an element $a$ in its universe, we
can define a structure $\Aa^*$ by removing the element $a$ and
expanding the structure with $2d$ additional unary relations that code 
all the various ways that elements could be connected to $a$.  That
is, for each PIF symbol $f$, we have unary relations $R^{\ra}_f$ and
$R^{\la}_f$ so that $R^{\ra}_f(b)$ if, and only if, $f(b) = a$ and $R^{\la}_f(b)$
if, and only if, $f(a) = b$.  It is now not difficult to translate the sentence
$\phi$ into a sentence $\phi^*$ so that $\Aa \models \phi$ iff $\Aa^*
\models \phi^*$.  In particular, $\phi^*$ asserts that $R^{\la}_f$
has at most one element.
\end{proof}

We now establish the relationship between $\PIFSpec_d$ and spectra
over \emph{graphs} of bounded degree.

\begin{defi}
Let  $\BDSpec_d$ denote the class of sets $S$ such that $S$ is a spectrum of
some sentence $\phi$ in a vocabulary with one binary relation $E$ and
some number of unary relations such that every model of $\phi$ interprets $E$
as a symmetric relation and has degree at most $d$. 
\end{defi}

The following theorem allows us to simplify formulae using lots of
symbols to formulae using only a single symmetric binary relation of
degree 3.

\begin{thm}\label{thsimplify}
If $S \in \PIFSpec_d$ then $\{2dn+l: n \in S\}$ is in $\BDSpec_3$ for
any $l$.
\end{thm}

\begin{proof}
Let $\phi$ be a formula using $d$ PIF symbols $f_1, \ldots, f_d$ such
that $S$ is the spectrum of $\phi$ and let  $\stra$ be a structure of
the same vocabulary.  Define the vocabulary $\tau$ to contain the
binary relation $E$ and $d$ additional unary predicates  $P_1, \ldots,
P_d$.  We define a $\tau$-structure $G$ by replacing each element $a$
of $\stra$ with a gadget consisting of a simple $E$-path of vertices
$p_1(a), \ldots, p_d(a), q_1(a), \ldots, q_d(a)$, where $p_i(a)$ and
$q_i(a)$ both satisfy $P_i$ and no other predicate from  $P_1, \ldots,
P_d$.  If $f_i(a)=a'$, there is an edge between $p_i(a)$ and
$q_i(a')$.  In addition $G$ contains $l$ isolated vertices.  It is
easily verified that $G$ has degree bounded by 3.

It is easy to write a formula $\phi_g$ of first-order logic whose
models are exactly the coloured graphs $G$ obtained in this way.  And,
it is also easy to see that $\stra$ can be interpreted in $G$.
Combining these, we get a sentence whose spectrum is exactly  $\{2dn+l:
n \in S\}$.
\end{proof}

We are also interested in spectra obtainable by models which are
planar graphs.  We write $\PPifSpec_d$ to denote the set of
$S \subseteq \bbN$ such that $S$ is a {\it planar} spectrum of some
formula $\phi$ using only unary relations and $d$ PIF symbols.  That
is to say, $n \in S$ iff there is some planar structure $\stra$ such
that $\stra \models \phi$.  Note that there is no obvious inclusion
either way between $\PPifSpec_d$ and $\PIFSpec_d$.  Some
$S \in \PIFSpec_d$ may be witnessed by a formula $\phi$ which in some
cardinalities has only non-planar models.  On the other hand, as
planarity is not itself definable in first-order logic, the set $S'$
of those $n$ for which $\phi$ has planar models of size $n$ may not be
itself the spectrum of a first-order sentence.  We write
$\FPPifSpec_d$ to denote the set of $S \subseteq \bbN$ such that $S$
is a spectrum of some formula $\phi$ using only unary relations and
$d$ PIF symbols, and such that all models of $\phi$ are planar
sets. $\FPPifSpec_d$ is a subset of both $\PIFSpec_d$ and
$\PPifSpec_d$.

\subsection{Complexity classes}

Our aim is to characterize spectra using complexity theory. For this
purpose, we define complexity classes of sets of numbers in terms of
resource bounds on machines accepting the binary representations of
numbers.  

We write $\binntime(f(N))$ to denote the class of sets
$S \subseteq \bbN$ such that there is a nondeterministic multi-tape
Turing machine that for any $N$ accepts the binary representation of
$N$ iff $N\in S$ and runs in time $O(f(N))$.  If we interpret the
classes $\NP$ and $\NE$ as sets of numbers in this way, then $\NP
= \bigcup_k \binntime(\log^k(N))$ and $\NE
= \bigcup_k \binntime(N^k)$.  By results of Fagin~\cite{fagin74} and
Jones and Selman~\cite{JonesS74}, we know that $\NE$ is exactly the
class of first-order spectra.  In the present paper we will relate
bounded-degree and planar spectra to the classes
$\binntime(N \log^k(N))$, which are between $\NP$ and $\NE$.  Indeed,
by the time hierarchy theorem for nondeterministic
machines~\cite{SFM78}, the classes are \emph{strictly} between $\NP$
and $\NE$.

We also write $\binntisp(T(N), S(N))$ to denote the class of sets
$S \subseteq \bbN$ such that there is a nondeterministic Turing
machine accepting the binary representation of $N$ iff $N \in S$ and
running in time $O(T(N))$ and space $O(S(N))$.  We write
$\binnts(f(N))$ to denote the union of $\binntisp(T(N), S(N))$ over
all functions $T$ and $S$ such that $T(N) \cdot S(N) = O(f(N))$.

These classes are less robust than the well known complexity classes
such as $\P$, $\NP$, or $\NEXPTIME$ because the functions bounding
resources are not closed under composition with polynomials.  The
classes are then sensitive to the exact computation model (e.g. single
tape Turing machine, an automaton with two stacks, multi-tape Turing
machine, RAM), and some of our results are stated for the specific
model of automata with two stacks, or equivalently, single tape Turing
machines which can insert or delete symbols from the tape; we use the
symbol $\binntime^S(f(N))$ instead of $\binntime(f(N))$ in this case.

However, the classes are somewhat robust in a weaker sense.  The
following is established by standard methods of speed-up theorems by
considering a larger tape alphabet.

\begin{prop}\label{palpha}
If a machine $M$ recognizes the binary representations of elements of
$S$ in time $O(f(N)) = \Omega(\log^3(N))$, then for any $\alpha > 0$
there is a machine $M'$ recognizing the same set and running in time
$\alpha f(N)$ for all $N > N_0$, for some $N_0 \in \bbN$.
\end{prop}

Proposition \ref{palpha} allows us to assume that a machine solving a
problem in $\binntime(f(N))$ actually solves it in time $\alpha f(N)$
for all $N > N_0$.  When relating complexity to spectra, the finitely
many exceptions $N \leq N_0$ are not important as a result of
Proposition \ref{propclosure}.

\section{Basic structure}

In this section, we describe a basic construction, consisting of the
description of a structure of each finite cardinality $N$, along with
a first-order axiomatisation of this class of structures.  The
structures in question are planar and of bounded degree and play an
important role in our results in encoding machine computations.

Let $\stra_N$ be the following structure over the vocabulary $\voc_0$
consisting of a single partial function symbol $f_{+1}$: 

\begin{itemize}
\item the universe $A_N = \{1, \ldots, N\}$,
\item $f_{+1}(n) = n+1$ for $n < N$,
\item $f_{+1}(N)$ is not defined.
\end{itemize}

\begin{thm}\label{makeorder}
There is a FO formula $\phi_M$ over vocabulary $\voc_1 \supset
\voc_0$ containing 2 PIF symbols, such that $\phi_M$ has models of
all non-negative integer cardinalities, and all models of $\phi_M$
restricted to $\{f_{+1}\}$ are isomorphic to $\stra_N$ for some
$N$. Moreover, all models of $\phi_M$ are planar. 
\end{thm}

\begin{proof}
The vocabulary $\voc_1$ consists of two PIF symbols $f_{+1}$ and
$f_{2x}$.  Instead of writing $f_{+1}(x)$ and $f_{2x}(x)$, we just
write $x+1$ and $2x$, respectively.  We also use $x+k$ for the $k$-th
iteration of $x+1$. 

The formula $\phi_M$ is a conjunction of the following statements:

\begin{enumerate}
\item \label{cnext} $x+1$ is defined for all values of $x$ except one.
  In the sequel, we write $N$ for the unique $x$ for which $x+1$ is
  not defined. 

\item \label{cprev} For all values of $x$ except one, there is an
  element $y$ such that $y+1=x$.  We write $x-1$ to denote this $y$.
  The unique element $x$ for which $x-1$ is not defined is denoted
  $1$.

\item \label{halving} If $2x$ and $x-1$ are both defined for a given
  $x$, then $(2(x-1)+1)+1$ is also defined, and $(2(x-1)+1)+1 = 2x$. 

\item \label{czero} We have $2 \cdot 1 = 1+1$.

\item \label{eitheror} Either $x=2y$ for some $y$, or $x+1=2y$ for some $y$.
\end{enumerate}

It is straightforward to check that $\stra_N$ expanded with function
$f_{2x}(x) = 2x$ is a model of $\phi_M$.  Figure~\ref{fig:spiral}
shows that such models are indeed planar.
\begin{figure}
\begin{center}
\includegraphics{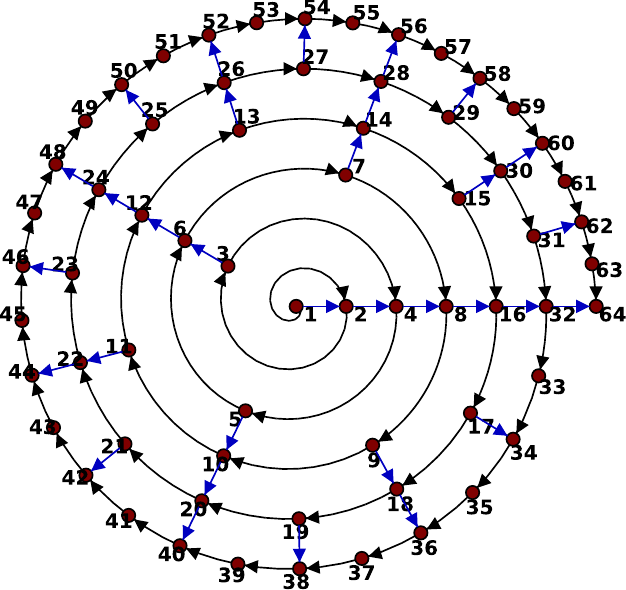}
\end{center}
\caption{Basic spiral structure.}
\label{fig:spiral}
\end{figure}

Next, we argue that, up to isomorphism, these are the only models of
$\phi_M$.  Let $\stra$ be a model of $\phi_M$. Suppose that the
interpretation of $f_{+1}$ has a cycle, i.e., there is an $a \in A$
such that $a+k = a$ for some $k > 0$.   Let $l$ be the minimum length
of such a cycle.  From condition~\ref{eitheror} we know that every
second element of this cycle is of the form $2x$ for some $x$.  In
particular, $l$ must be even.  By condition~\ref{halving} and the
injectivity of $f_{2x}$ we know that the elements $x$ such that $2x$
is on the cycle must themselves form a cycle, this time of length
$l/2$. This contradicts the assumption that $l$ is minimal. 

Since there are no cycles, it follows from conditions~\ref{cnext}
and~\ref{cprev} that each element of $A$ is of form $1+(k-1)$ for some
integer $k$.  Writing $a_k$ for this element, we can prove inductively
that $2(a_k) = a_{2k}$.  This means that $a_k \mapsto k$ is the
required isomorphism to $\stra_N$ (where $N= |A|$).
\end{proof}

The spectrum of $\phi_M$ is the set of all natural numbers which is
not itself very interesting.  However, the structures we have
axiomatized can be used to describe simple properties of numbers in a
simple way.

\begin{exa}\label{expowersoftwo}
Consider the vocabulary $\voc_1$ expanded with an additional unary
relation $P$, and add an axiom: $P(x)$ iff $x=1$ or $\exists y \, P(y)
\wedge x=2y$.  We can prove inductively that in any expansion of
$\stra_N$ that is a model of this axiom, $P$ has to
be interpreted as the set of powers of two.  Adding a further axiom
$P(N)$ we obtain a sentence whose spectrum is the set of powers of two.
Moreover, all models of this sentence are planar graphs.

We could also achieve a similar effect by adding a function $x \ra
2^x$ to our vocabulary, and the following axiom: $2^{x+1} = 2 \cdot 2^x
\wedge 2^1 = 1+1$.  We can prove inductively that this symbol has to
be interpreted in the intended way.  However, the models are no longer
planar, and as we add a new PIF, the degree bound increases to 6.
\end{exa}

\begin{exa}\label{exmultiply}
Consider the vocabulary $\voc_1$ expanded with an additional constant
symbol $C$ and two additional functions that we denote $x+C$ and $x
\cdot C$.  We axiomatise these PIFs by means of the following
axioms: $1+C = C+1$, $(x+1)+C = (x+C)+1$; $1 \cdot C = C, (x+1) \cdot
C = x \cdot C+C$.  We can prove inductively that $x+C$ and $x \cdot C$
are interpreted in the intended way.  By adding the following axioms:
$C \neq 1, \exists y\ y \neq 1 \wedge N=y\cdot C$,  
we obtain a sentence whose spectrum is the set of composite numbers,
and all of whose models are of degree at most 8.
\end{exa}

\begin{exa}\label{exfib}
Consider the vocabulary $\voc_1$ expanded with an additional PIF $F$
axiomatized by the following axioms: $F(1) = 2$, $F(n+1) = F(n)+2$ if
$\exists m F(m)=n$, and $F(n+1) = F(n)+1$ if no such $m$ exists.  We
also have an additional unary relation $\Phi$ and the axioms:
$\Phi(x)$ iff $x=1$ or $\exists y\ x=F(y) \wedge \Phi(y)$; and
$\Phi(N)$.  It can be shown (by induction again) that the spectrum of
the resulting sentence is the set of Fibonacci numbers.

Our vocabulary here has three PIFs, so the structures have degree
bounded by 6.  However, the function $2x$ is not necessary in this
particular case.  We could add the symbols $F$ and $\Phi$ to the
vocabulary $\voc_0$ instead.  We need an axiom similar to
condition~\ref{eitheror} in the definition of $\phi_M$, but for $F(x)$
instead of $2x$.  With this $F(x)$ is sufficient to fix the structure.
By removing the axioms related to the function $2x$, we obtain a
sentence whose spectrum is again the set of Fibonacci numbers, and all
of whose models are planar of degree at most 4.
\end{exa}

\begin{exa}\label{exbinary}
Consider $\voc_1$ expanded with a function symbol $P$ axiomatized
with the following axioms: $P(1) = N$; if
$P(x)\neq 1$ then $P(x) = 2P(x+1)$ or $P(x) = 2P(x+1)+1$; and if
$P(x+1)$ is defined then $P(x)$ is also defined.
Let $\phi_1(x)$ denote the value $1$ if $P(x) = 2P(x+1)+1$ and $0$ if
$P(x) = 2P(x+1)$.  Let $l$ be the unique element such
that $P(l)=1$ and let $\phi_1(l) = 1$.  Then $\phi_1(l)\ldots \phi_1(1)$ is the binary
representation of $N$. 
\end{exa}

\begin{exa}
Let $S \in \binnts(N)$.  We show that $S \in \PIFSpec_4$.  This is a
consequence of Theorem~\ref{mainlower} below, but we include the
argument here as it is strikingly simple.

Example~\ref{exbinary} allows us to axiomatise structures $\stra_N$
expanded with the binary representation of $N$.
Example~\ref{exmultiply} then shows how to axiomatise a grid of width
$C$ and height $C_2$, where $C \times C_2 \leq N$ (we only use the
functions $x+1$ and $x+C$).  Let $C = S(N)$ and $C_2 = T(N)$.  Then on
the grid, with additional unary relations, we can describe the
computation of a Turing machine in the standard way.
\end{exa}

\section{Lower bound: Turing machines via spectra}

In this section we establish a lower bound, by showing that the class
of bounded-degree spectra includes all sets in $\binntime(N)$.  The
proof proceeds by constructing, from a Turing machine $A$, a
first-order sentence whose models are, in a sense we make precise,
codings of accepting computations of $A$.  For simplicity of
exposition, in the proof we confine ourselves to machines with one
tape, but the construction can easily be extended to machines with two
or more tapes.

\begin{thm}\label{mainlower}
$\binntime(N) \subseteq \PIFSpec_6$.
\end{thm}

\begin{proof}
Suppose $S \subseteq \bbN$ is in $\binntime(N)$ and there is a single
tape Turing machine $A$ that accepts the binary
representation of $N$ just in case $N \in S$ and runs in time $O(N)$.
By Proposition~\ref{palpha}, we can assume that the running time $T$ of
$A$ when the input is the binary representation of $N$ is such that
for large enough values of $N$, $5T + 2 \log T \leq N$.

Our aim is to describe the construction of a first-order sentence
$\phi_A$ which has a model of size $N$ if, and only if, the binary
representation $w_N$ of $N$ is accepted by $A$.  The vocabulary includes a
unary relation $R$ and the idea is that in any model of $\phi_A$ the
substructure defined by $R$ is an encoding of an accepting computation
of $A$ on $w_N$.  The aim is to define this encoding in such a manner
that the size of this encoding is bounded by a linear function of $T$,
the running time of $A$.  In particular, we will see that  $5T + 2
\log T \leq N$ suffices.  

Note that it is straightforward, by axiomatising a suitable $T \times
T$ grid, to give such an encoding that is quadratic in $T$.  That is,
we can axiomatise a model in which the elements are pairs $(x,y)$
representing a tape cell $x$ at time $y$ with unary relations coding
the state and tape contents of the machine.  The
challenge here is to give a linear size construction and we achieve
this by having single elements corresponding to a tape cell $x$ over an
interval of time during which the content of the cell does not change.
This is similar to the idea behind the construction
in~\cite{Gra84,Gra85} which is used to show that any set $S \subseteq
\mathbb{N}$ that can be decided in time $O(N)$ on a nondeterministic
random-access machine is the spectrum of a sentence with unary
function symbols.

Our vocabulary contains: three PIFs $\lfa$, $\rfa$, and $\ufa$; the
unary relation $R$; three additional unary relations  $M$,
$\lma$, $T_\rightarrow$; a unary relation $S_q$ for each state
$q$ of $A$ and a unary relation $L_c$ for each letter $c$ of the
tape alphabet of $A$.

We now describe the intended models of $\phi_A$.  These are encodings
of accepting computations of $A$ on an input $w_N$.  It should be
noted that the encoding is not necessarily unique, but $\phi_A$ does
guarantee that it is a valid encoding and it is of size $N$.

In a model of $\phi_A$, the intention is that each element $e$ in the
relation $R$ represents a pair $x,[y_1, y_2)$ where $x$ is a position
on the tape and $[y_1,y_2)$ (with $y_1 < y_2$) is an interval of
time in the computation.  The meaning of the unary relation symbols,
other than $R$ can then be given as follows.  We have $M(e)$ if, and
only if, $y_2 = y_1 +1$ and at time $y_1$, the machine is reading
position $x$ in the tape; $T_\leftarrow(e)$ if, and only if, position $x$ is to the
left of the position being read by the machine during the time
interval $[y_1,y_2)$ and  $T_\rightarrow(e)$ if, and only if, position $x$ is to the
right of the position being read by the machine during the time
interval $[y_1,y_2)$; $L_c(e)$ if, and only if, the tape cell $x$ during
the interval  $[y_1,y_2)$ contains the letter $c$; and $S_q(e)$ if,
and only if, $M(e)$ and at time $y_1$ the machine is in state $q$.

The PIFs $\lfa$, $\rfa$, and $\ufa$ then connect these elements.  If an
element $e$ corresponds to position $x$ and interval $[y_1, y_2)$,
then $\ufa(e)$ corresponds to position $x$ and interval $[y_2, y_3)$,
for some $y_3$ with 
$y_3 > y_2$.  Thus, $\ufa$ is used to build a list of contents of
the tape at position $x$ during consecutive time intervals.  We
write $\uan$ as an abbreviation for $\ufa^{-1}$.  This enables us to
enforce the requirement that the contents of a tape cell do not change in
between successive visits of the tape head: if $L_c(e)$ holds and
$M(e)$ does not hold, then $L_C(\ufa(e))$ must hold.

The $e$ is an element corresponding to a pair $x,[y_1, y_2)$ and
$\lfa$ is defined at $e$, then $\lfa(e)$ corresponds to a pair
$x-1,[y_1,y_2')$ for some $y_2' \geq y_2$.  That is, $\lfa$ points to
the element on the tape immediately to the left for a time interval
starting at the same time and extending at least as far.  The function
$\lfa$ is defined only at elements $e$ corresponding to positions at
or to the left of the tape head.  That is, either $M(e)$ or  $\lma(e)$
must hold for $\lfa(e)$ to be defined.  The idea is that since $e$
corresponds to a position on the tape to the left of the tape head,
the tape head must pass through this position to get to to the
position immediately to its left.  Thus, the position immediately to
the left is unchanged for at least as long.
Moreover, we require that when $\lfa(e)$ is
defined, and corresponds to the pair $x-1,[y_1,y_2')$ with $y_2' >
y_2$, then $\ufa(e)$ corresponds to the pair $x,[y_2,y_2')$.  Thus,
from $e$, we get to the position immediately to the left during the
same interval, either by taking $\lfa(e)$ if it is defined or
$\lfa(\uan(e))$ otherwise.  For ease of presentation, we introduced
some abbreviations.  We write  $\lmt(e)$ as a shorthand for $M(e) \vee
\lma(e)$, i.e.\ to denote that $e$ represents a position at or to the
left of the tape head. We also write $\Lfa(e)$ to
denote the element  $\lfa(e)$ if it is defined, and $\lfa(\uan(e))$
otherwise, i.e.\ the one that encodes the tape position immediately to
the left of the one represented by $e$, and during the same time
interval. 

The function $\rfa$ is to be interpreted symmetrically, taking an
element $e$ corresponding to the pair  $x,[y_1, y_2)$ to an element
  coding the pair  $x+1,[y_1,y_2')$ for some $y_2' \geq y_2$, provided
    that  $M(e) \vee \rma(e)$ holds.  Again, we 
 write $\rmt(e)$ as a shorthand for $M(e) \vee \rma(e)$ and we write $\Rfa(e)$ to
denote the element  $\rfa(e)$ if it is defined, and $\rfa(\uan(e))$
otherwise.

We assume that there are two special elements marking the first and
last positions on the tape and that the associated time intervals
cover the whole computation (i.e.\ these tape cells are not changed).
In what follows, we refer to these two special elements as $l$ and $r$
respectively. 

We now wish to express the following condition we refer to later as
$(\star)$:
\begin{quote}
Suppose that $M(e)$ and  $\ldots, \Lfa^2(e), \Lfa(e), e, \Rfa(e),
\Rfa^2(e), \ldots$, along with the relations $S_q$ and $L_c$ on these
elements,   describes a configuration of $A$.  Then for some $e'$
($e'=\ufa(\Rfa(e))$ if in this configuration $A$ moves right, and
$e'=\ufa(\Lfa(e))$ if $A$ moves left), the sequence $\ldots, \Lfa^2(e'),
\Lfa(e'), e', \Rfa(e'), \Rfa^2(e'), \ldots$ describes the next
configuration of $A$.
\end{quote} 
The following axioms ensure this property.  In all these, quantifiers
are relativized to the set $R$.

\begin{itemize}
\item $M, \lma, \rma$ is a partition of $R$;
\item $\{L_c: c \in \Sigma\}$ is a partition of $R$;
\item $\{S_q\}$ is a partition of $\{e: M(e)\}$;
\item if $\rma(e)$ then $\lfa(e)$ is not defined, otherwise $\Lfa(e)$ is
  defined except when $e=l$; similarly, if $\lma(e)$ then $\rfa(e)$ is
  not defined, otherwise $\Rfa(e)$ is defined, unless $e=r$;
\item the tape content does not change except under the tape head:
if $L_c(e)$ and not $M(e)$ and $\ufa(e)$ is defined then $L_c(\ufa(e))$;
\item correctness of the right side of the tape (height 1):
if $M(e)$ or $\rma(e)$, and $\rma(\ufa(e))$, and $\rfa(e)$ is defined, and $\rfa(\ufa(e))$ is defined, then $\ufa(\rfa(e)) = \rfa(\ufa(e))$;

\item correctness of the right side of the tape (height 2):
if $M(e)$ or $\rma(e)$, and $\ufa(\ufa(e))$ is defined, and $\rfa(e)$ is
defined, and $\rfa(\ufa(e))$ is not defined, then $\ufa(\rfa(e)) = \rfa(\ufa(\ufa(e)))$,

\item no greater height: if $\ufa(e)$ is defined and $M(e)$ or $\rma(e)$, then $\rfa(e)$ or $\rfa(\ufa(e))$ is defined;

\item symmetric axioms for correctness of the left side of the tape;

\item time zero: there is only one $e$ such that $M(e)$ and $\uan(e)$ is not defined; if $\uan(e)$ is not defined, then $\uan(\rfa(e))$ and $\uan(\lfa(e))$ are also not defined;

\item end of computation: there is only one $e$ such that $M(e)$ and
  $\ufa(e)$ is not defined and  for this $e$, $S_q(e)$ for some final
  state $q$;  if $\ufa(e)$ is not defined and $\Lfa(e)$ is defined then $\ufa(\Lfa(e))$ is not; if $\ufa(e)$ is not defined and $\Rfa(e)$ is defined then $\ufa(\Rfa(e))$ is not;

\item correctness of computation: if $M(e)$ and $S_q(e)$ and $q$ is
  not a final state, then either $M(\ufa(\Rfa(e)))$ and $\ufa(e) =
  \lfa(\ufa(\Rfa(e)))$, or $M(\ufa(\Lfa(e)))$ and $\ufa(e) =
  \rfa(\ufa(\Lfa(e)))$; moreover, if $L_c(e)$ and $L_{c'}(\ufa(e))$, then
  there is a valid transition of $A$ in state $q$ to replace $c$
  by $c'$ and move left or right as appropriate.
\end{itemize}

The construction is illustrated in Figure~\ref{tmsim}.  In this
figure, the large red circles depict elements $e$ for which $M(e)$
holds while the smaller green circles represent other elements $e$.
The upward pointing arrows depict the partial function $\ufa$, the
rightward pointing arrows depict the partial function $\rfa$ and the
leftward pointing arrows depict the partial function $\lfa$.
\begin{figure}
\includegraphics{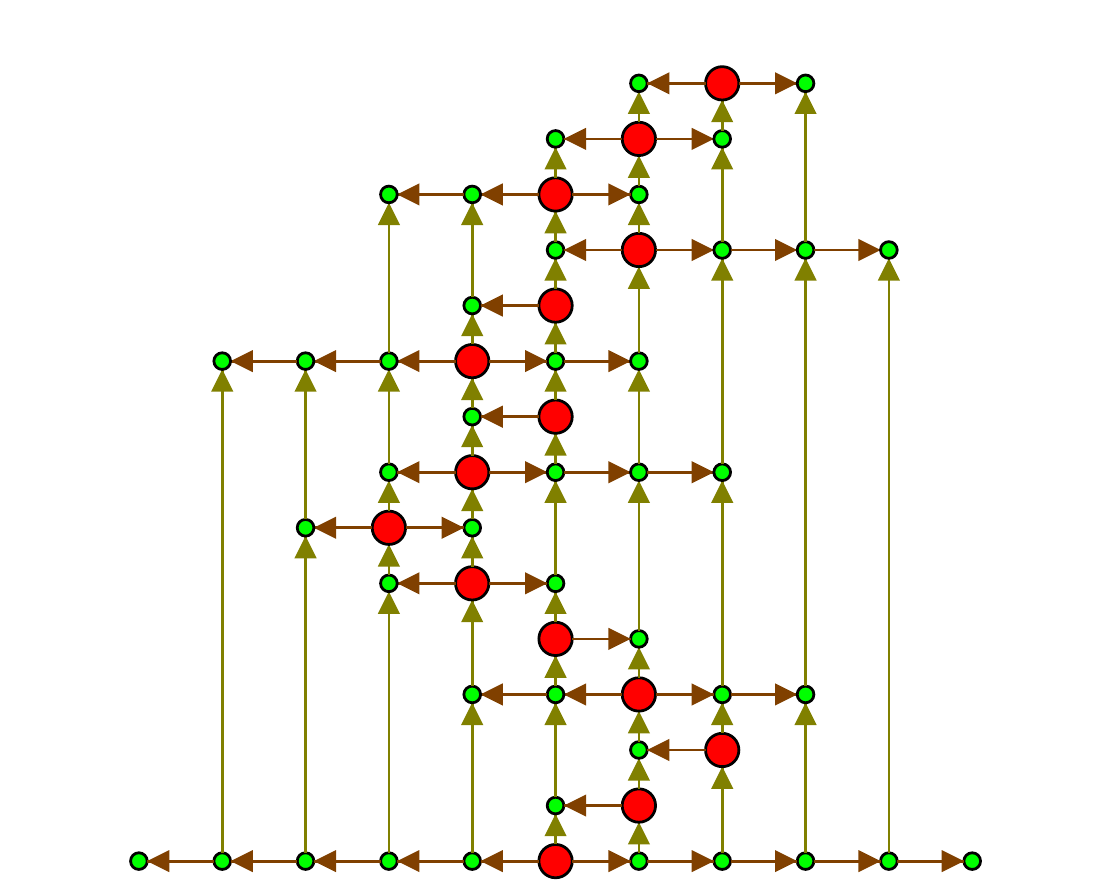}
\caption{Model coding the computation of a machine.}
\label{tmsim}
\end{figure}

In addition to the above axioms (which ensure condition 
 $(\star)$), we can use the construction from  Example~\ref{exbinary}
to include in $\phi_A$ the condition that the initial contents of the
tape are the binary encoding of the size of the model itself.  This
ensures that if $\phi_A$ has a model of size $N$ then $A$ accepts the
binary encoding of $N$.  The vocabulary contains in all six PIF
symbols: three from the construction above, two required by the
construction  from Theorem~\ref{makeorder} and one additional one for
the construction in  Example~\ref{exbinary}.

Now, we need to show that if the binary encoding $w_N$ of $N$ is accepted by
$A$, then there is a model of $\phi_A$ of size $N$.  Suppose that the
accepting computation of $A$ on input $w_N$ takes time $T$ and space
$S$.  We show that we can construct a model satisfying the axioms
above with at most $S + 4T + 2\log T$ elements.  Since we can assume
$S \leq T$ and by assumption $5T + 2 \log T \leq N$, this model has
fewer than $N$ elements.  We can then obtain a model of $\phi_A$ with
exactly $N$ elements by adding additional elements that are not
included in $R$.

We construct, from the accepting computation of $A$, a model of the
appropriate size by the following iterative procedure:

\begin{itemize}
\item We have $S + 2 \log T$ elements coding the initial contents of
  the tape, along with $\log T$ blank cells on either side of it.
  These are connected through the functions $\lfa$ and $\rfa$.

\item For each subsequent configuration of $A$, we add a layer of
  elements, which are connected to previous layers using $\ufa$.  The
  new layer will always contain two elements corresponding to the positions
  in the tape where the tape head was in the previous configuration
  and the position where it is in the new configuration.  In addition
  to these two, we add new elements to the layer by the following
  rule: if $e$ has been added to the new layer,  $\rmt(e)$ holds and
  $\Rfa(e)$ is not defined, then we add a new element which will be
  $\rfa(e)$.  Similarly, if  $\lmt(e)$ holds and $\Lfa(e)$ is not
  defined, we add an element to be $\lfa(e)$.
\end{itemize}
This is illustrated in  Figure~\ref{tmsim} which was obtained
by running this procedure.  In this case, $S=5$ and $T=14$, but two
cells of padding on either side suffice.

To obtain a bound on the size of this model, we note that the bottom
layer contains $S + 2\log T$ elements.  We want to show that the total
number of elements in the additional layers is bounded by $4T$.  We
cannot put an absolute bound on the number of elements added in each
layer, but we use an amortized analysis to obtain the required bound.
For each element $e$ for which $M(e)$ holds, let the potential
$\Phi(e)$ be defined to be  the number of elements in the sequence $e,
\Rfa(e), \Rfa^2(e), \ldots$ where $\rfa(e)$ is not defined, plus the
number of elements in the sequence $e, \Lfa(e), \Lfa^2(e), \ldots$ where
$\lfa(e)$ is not defined.   If we add the number of new elements
created in a layer to the change of potential $\Phi(e)$, we can see
that the result will be always at most 4.   Initially, the potential
is 2. Thus the number of positions added in $T$ steps is at most
$4T$.  It can be also easily checked that the padding of length $\log
T$ at each end of the tape is always sufficient.
\end{proof}

The construction in the proof can also be understood as simulating not
a single tape machine, but an automaton with two stacks (given by
$\lma$ and $\rma$).   It is straightforward to generalize this
construction to more stacks, or equivalently, more tapes. 

\begin{cor}\label{simpcor}
$\binntime(N) \subseteq \BDSpec_3$.
\end{cor}

\begin{proof}
Let $S \in \binntime(n)$. For each $l = 0, \ldots, 11$, let $S_l = \{n:
12n+l \in S\}$. The set $S_l$ is also in $\binntime(n)$. From Theorem
\ref{mainlower}, $S_l$ is in $\PIFSpec_6$. From Theorem
\ref{thsimplify}, the set $S_l' = \{n: n \in S \wedge n \bmod 12 = l\}$ is
in $\BDSpec_3$. Since a union of spectra is also a spectrum (if 
$\spec(\phi_i) = S_i,$ then $\spec(\bigvee \phi_i) = \bigcup S_i$), 
also $S \in \BDSpec_3$.
\end{proof}

The construction we have given uses a large number of unary
relations.  Although Corollary~\ref{simpcor} shows that we can reduce
the degree of the models to 3, it is unclear whether we can also
reduce the number of unary relations required.

\section{Upper bound: spectra via Turing machines}

In this section we aim to establish an upper bound on the class of
bounded degree spectra.  To be precise, we show that any such spectrum
can be recognised by a nondeterministic Turing machine in time $O(N
(\log N)^2)$.  One consequence of this is that there are spectra that
are not bounded-degree spectra. 

To establish the result, we need to show that for any first-order
sentence $\phi$, we can construct a nondeterministic machine $A_\phi$
that given a positive integer $N$ as input (it does not much matter
whether $N$ is given in unary or binary as we allow a running time
that is greater than the value of $N$) will decide whether or not
$\phi$ has a model with exactly $N$ elements.  We assume, for
simplicity, that $\phi$ is in a vocabulary with one binary relation
$E$ and a number of unary relations $R_1, \ldots, R_k$.  We can think
of structures in this vocabulary as coloured graphs and we use the
language of graphs to describe them. 

The machine $A_\phi$ proceeds by nondeterministically guessing a graph
$G$ with $N$ vertices and degree bounded by $d$ and then verifying
that $G$ is indeed a model of $\phi$.  The algorithm for deciding
whether $G$ satisfies $\phi$ relies on Hanf's locality theorem~\cite{hanf}
(see also~\cite[Theorem~4.12]{Lib04}).  For a vertex $v$ in $G$ and a positive integer
$r$, we write $N_r(v)$ to denote the substructure of $G$ induced by
the set of vertices with distance at most $r$ to $v$, with a
distinguished new constant interpreted by $v$ itself.  Note that if
the degree of $G$ is at most $d$, then $N_r(v)$ has at most $1
+d(d-1)^{r-1}$ vertices in it.  We write $S_r$ to denote the number  $1
+d(d-1)^{r-1}$.  We denote by $\tau(v)$ the \emph{isomorphism type} of
the structure $N_r(v)$.  To be precise, we can take $\tau(v)$ to be
a canonical structure isomorphic to $N_r(v)$ on the domain
$\{1, \ldots, n_\tau\}$ where $n_\tau = |N_r(v)|$.  Since $n_\tau \leq
S_r$, there are only finitely many distinct types.  We write $\Tt$ for
the set of all types.  For a fixed positive integer $M$, let
$f_{r,M}(G): \Tt \ra \{0,\ldots,M\}$ be the function that assigns to
each type $\tau \in \Tt$ the minimum of $M$ and the number of vertices
$v$ in $G$ with $\tau(v) = \tau$.  For a pair of graphs $G$ and $H$,
we write $G \sim_{r,M} H$ if $f_{r,M}(G) = f_{r,M}(H)$.
 Then, Hanf's locality theorem can
be stated as follows.
\begin{thm}[Hanf's locality\cite{hanf}]\label{hanf}
Let $\phi$ be a FO formula.  Then there exist integers $r$ and $M$
such that, if $G \sim_{r,M} H$ then $G\models \phi$ if, and only if
$H \models \phi$.
\end{thm}

Now, we describe how this can be used to construct the machine $A_\phi$
as required in the proof of the main theorem.

\begin{thm}\label{mainupper}
$\BDSpec_d \subseteq \binntime(N \log^2 N)$.
\end{thm}
\begin{proof}
Given a first-order sentence $\phi$, we describe a multi-tape non-deterministic
machine $A_\phi$ which given an input $N$ decides if $\phi$ has a
model with exactly $N$ elements.  

First, $A_\phi$ non-deterministically guesses a graph $G$ on $N$
vertices by writing its description on tape $A$.  The description of
$G$ consists of a list of $N$ \emph{vertex descriptions}, where each vertex
description consists of an identifier of a vertex $v$ (which is an
integer given by at most $\log N$ bits), a list of the unary
predicates satisfied by $v$ and a list of the identifiers of the
neighbours of $v$.  Note that each vertex description is $O(\log N)$
bits long, because the number of neighbours of $v$ is bounded by $d$.
Thus, the total length of the description of $G$ is $O(N \log N)$.  We
can assume that the vertex identifiers are exactly $\{1,\ldots,N\}$
and the vertex descriptions are enumerated in increasing order of the
identifier. 

Secondly, the machine $A_\phi$ non-deterministically guesses the type
$\tau(v)$ of each vertex.  To be precise, it writes on tape $B$, for
each vertex $v$ the type $\tau(v)$ (as there are only a bounded number
of types, this can be specified in a constant amount of space) and a
list of identifiers of vertices of $G$ that correspond to the elements
$\{1, \ldots, n_\tau\}$.  Note that the entire list has length
$O(N \log N)$.  To complete its task, $A_\phi$ needs to verify two
things: (1) that the guessed type of each vertex $v$ on tape $B$ is
consistent with the graph $G$ described on tape $A$; and (2) that the
list of types on tape $2$ does ensure that $G \models \phi$.  Since
the list of types determines $f_{r,M}$, and the latter determines
whether $G\models \phi$, (2) is a simple table look up.  We now
describe how $A_\phi$ can, non-deterministically, in time $O(N \log^2
N)$ perform task (1).

The essential idea is that we produce on a series of $S_r$ additional
tapes copies of the vertex descriptions indexed in a suitable way so
that we have available for each vertex $v$, not only the vertex
identifiers of the vertices in its neighbourhood but their full vertex
descriptions. 
That is, tape $i$ will contain for each vertex $v$, a complete vertex
description of the $i$th vertex that occurs in $\tau(v)$ and these
will be enumerated in increasing order of $v$.  To achieve this, we
first scan tape $A$ and for each vertex description $v$ we encounter,
we guess the number $k$ (which may be $0$ and is at most $S_r$) of vertices $w$ such
that $v$ is vertex number $i$ in the description of $w$, and copy the
vertex description of $v$ on to tape $i$ $k$ times.  We then
non-deterministically generate a permutation of the vertex
descriptions on tape $i$.
 This can be done in time proportional to the length of tape $i$ times
 $\log(N)$, which is $O(N \log^2(N))$.  This is done by guessing a
 subset $S$ of the set of vertex descriptions, splitting the list into
 elements of $S$ and non-elements of $S$ (without changing order), and
 merging them again so that all elements of $S$ are after all
 non-elements of $S$.  Any permutation can be obtained through
 $\log(N)$ such operations.  This sorting technique is similar to the
 radix sort algorithm \cite[Sec.~9.3]{CLR89}.  Finally, we check that the
 permutation generated is indeed correct by making an additional scan
 of tape $A$ and checking that the vertex identifier that occurs in
 position $i$ in the type description of vertex $v$ matches that of
 the vertex description that occurs in position $v$ on tape $i$.

Tapes $1$ through $S_r$ now provide a suitably indexed table of vertex
descriptions and it is straightforward to check that the type
$\tau(v)$ ascribed to each vertex $v$ on tape $B$ is consistent with
the graph description on tape $A$.  Indeed, as we simultaneously scan
tapes $1$ through $S_r$, at position $v$ we have access to the full vertex
descriptions of all vertices occurring in the vertex description of
$v$ and we can check that $N_r(v)$ is indeed isomorphic to $\tau(v)$.
\end{proof}

We know, as a result of the time hierarchy theorem for
non-deterministic machines that $\binntime(N \log^2 N)$ is a proper
subset of $\NE$.  Since the latter is the class of all spectra, it
follows from Theorem~\ref{mainupper} that there are spectra that are
not in $\BDSpec_d$ for any $d$.

The proof of Theorem~\ref{mainupper} is based on the construction due
to Seese~\cite{See96} that shows that, for any fixed first-order
sentence $\phi$, the problem of deciding whether a given graph on $N$
nodes and degree bounded by a constant $d$ satisfies $\phi$ is
solvable in time linear in $N$ by a deterministic random-access
machine.  It follows that determining whether $\phi$ has a model that
is a graph on $N$ vertices with degree bounded by $d$ can be done by a
nondeterministic random-access machine in time $O(N)$.  Such a machine
can first guess the description of such a graph in time $O(N)$, since
the degree bound ensures that the graph has a description whose length
is bounded by $O(N)$.  Our proof aims at reconstructing this argument
in the context of multi-tape Turing machines.  This can also be
established by a general result due to Monien~\cite{Monien77} which
shows that any set decidable in linear time by a nondeterministic
random-access machine is in $\binntime(N \log^2 N)$.

\section{Assuming planarity}\label{secplanar}

In this section we investigate to what extent upper and lower bounds
similar to those in the previous two sections can be established in
the case where we only consider planar structures.  That is, we are
interested in the \emph{planar spectrum} of a formula $\phi$---the set
of those integers $n$ such that there is a model of $\phi$ of size $n$
which is also planar.

As we noted earlier, the models of the formula $\phi_M$ constructed in
the proof of Theorem~\ref{makeorder} are all planar.  Moreover,
Example~\ref{expowersoftwo} yields a formula whose spectrum is the set
of powers of two, and again all models are planar.  The more
complicated examples of Example~\ref{exmultiply} and~\ref{exbinary}
do not yield planar models.  However, Example~\ref{exfib} gives
us a simple formula whose planar spectrum is the set of Fibonacci numbers.

The construction of a model encoding the computation of a Turing
machine that we used in the proof of  Theorem~\ref{mainlower} yields a
planar graph as long as the machine only uses a single tape (or more
precisely two stacks).  But, the proof also relies on encoding in the
formula the statement that the initial contents of the tape encode (in
binary) exactly the size of the structure.  This relies on the
construction in Example~\ref{exbinary} and the models are no longer
planar.  Instead, we can modify the construction so that the size of
the model is determined by the length $S$ of the tape and the number
$T$ of steps taken.  Thus, we can easily modify the construction so
that at each time step, exactly four additional elements are used, by
spreading the amortized cost in the proof of Theorem~\ref{mainlower}
over the time steps.  Specifically, we axiomatize the cases where the
potential $\Phi$ is reduced and ensure there are additional dummy
elements in these cases.  Thus, as long as all elements in the
relation $R$ are part of the coding of the computation, there will be
exactly $S+4T$ of them.  As the initial tape contents in the machine
form the bottom line (see Figure~\ref{tmsim} above), we can add
gadgets below this line, without violating planarity, which allow us
to formulate statements in the formula about the length $S$ and the
contents of the tape.

This idea relies on the assumption that everything in $R$ forms part
of the machine computation.  This is tricky to enforce as illustrated
by Figure~\ref{tmsim-perio}.  Note that this picture has been
rotated 90 degrees in order to fit the page and now the time steps run
from left to right and a single column connected vertically represents the tape contents.  The picture illustrates the coding of a
machine that alternately moves six steps to the right on the tape and then
six steps to the left.

\begin{figure}
\includegraphics{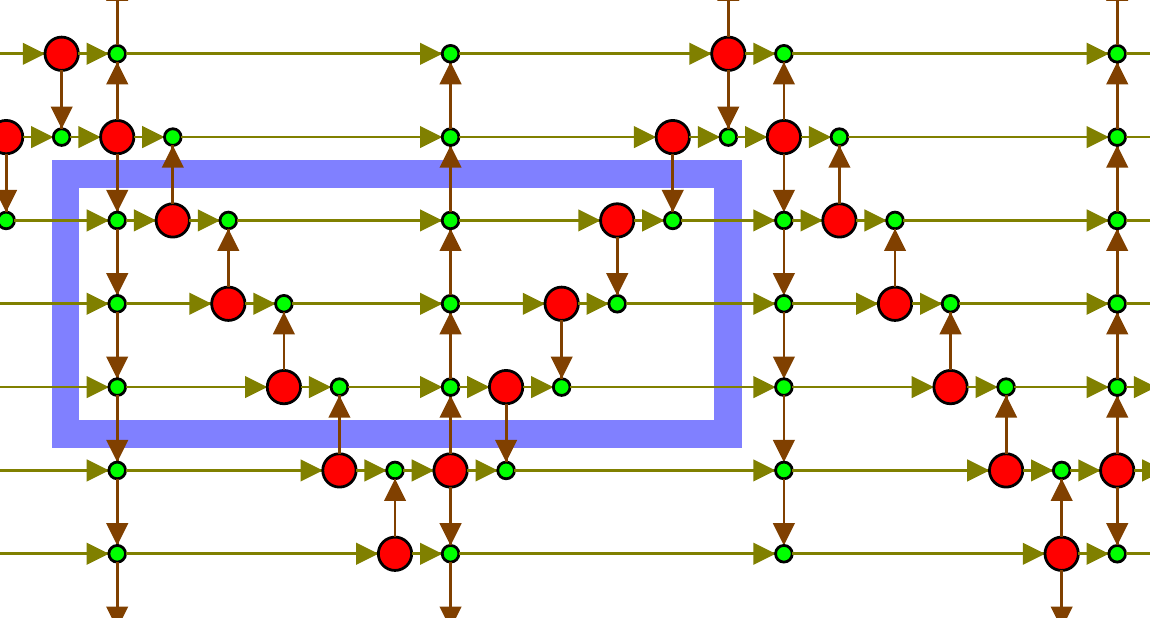}
\caption{Machine simulation with section for tiling.}
\label{tmsim-perio}
\end{figure}
In this picture, the area bounded by the blue rectangle defines a
subgraph which can be used to tile the plane.  That is to say, we can
take infinitely many copies of this subgraph and connect them by
identifying the edges leaving at the top with the incoming edges in
the copy above and similarly to the left and right.  This yields an
infinite graph $G$ which satisfies all the axioms apart from the edge
cases which talk of elements that appear only once in the coding of
the computation.  Hence, the disjoint union of a large enough model of
the axioms with $G$ still yields a model (albeit an infinite one).
Furthermore if, instead of tiling the plane, we use the rectangle to
tile a torus, we obtain a finite graph any number of copies of which
can be adjoined to a model to obtain a valid model.  This means that
the spectrum of our sentence must be ultimately periodic.  However,
the embedding of the tiling on a torus is necessarily not planar and
restricting ourselves to planar graphs, it turns out that we can
ensure that all elements in $R$ really are part of the coding of the
computation and this allows us to axiomatise valid computations.  In
the next two subsections, we carry out two such constructions, for two
different classes of machines.  In terms of complexity, the two
results are incomparable.

\subsection{A Turing machine with a clock}

\begin{thm}\label{clocksim}
$\binntime^S(N/\log N) \subseteq \PPifSpec_3$.
\end{thm}                                 
\begin{proof}

Let $M$ be a machine such that for any $N \in \bbN$, it accepts or
rejects the binary representation of $N$ in time $O(N/\log N)$.

We consider a machine $M'$ that takes as input a word of the form 
$w \# u \# b_C \# b_N$, where $w,u \in \{0,1\}^*$ and
$b_C$ and $b_N$ are
binary strings representing integers $C$ and $N$ respectively.
The machine $M'$ performs the following tasks:

\begin{itemize}
\item Replace $w$ on the tape with $b_s$ , the binary string representing the number $s$ of occurrences of $1$ in the string $w$.  This can be done in time $\Theta(|w| \log s)$.
After this step, 
the tape contains $b_s \# u \# b_C \# b_N $.

\item Based on the values of $s$, $u$, and $C$, 
verify that $N$ contains exactly the size of the model. This point will be
explained later. We also erase everything
from the tape, except $\# b_{C} \# b_N$.
This can be done in time polynomial in $s+|b_C|+|u|$.

\item Verify that the machine $M$ accepts the word $b_N$ in exactly $C$ steps.
This is done by simulating the machine $M$, while keeping the string between the
$\#$ signs as a counter that is decremented at each step and which is shifted left or
right as necessary along with the head movements of $M$. This can be done in
time $\Theta(|b_{C}| C)$.
\end{itemize}

Now, we consider representing the computation of the machine $M'$ as a
structure of size exactly $N$.  The structure consists of the following parts.

\begin{itemize}
\item The \emph{base line} which  is a straight horizontal path of length $S$ representing the initial tape contents.
We let $S = 3 \cdot 2^s$ for some
$s$, and partition the initial tape into three equal length segments,
the first and third of which are blank and the middle one contains
the input $w \# u \# b_C \# b_N$, padded with blanks on the right side.  Here  $w$ is the word of length $2^{s-1}$ which contains $1$s at positions
which are powers of two.

\item The computation, represented as a subgraph of size exactly $4T$ above the base
line, where $T$ is the number of steps taken. We use the construction from
Theorem \ref{mainlower}, improved as outlined above so that the size of the subgraph
is exactly $4T$.

\item To enforce the length requirement, the spiral structure
from Example~\ref{expowersoftwo} is attached below the base line,
but now requiring that the outer
edge of the spiral is $3\cdot 2^s$, instead of checking whether the
whole spiral is a power of $2$. Additionally, the spiral structure
is used to enforce that the occurrences of $1$ in $w$ are  indeed at the powers of two.
There are exactly $S$ elements in the spiral structure.

\item One additional element is attached to each occurrence of $1$ in $u$
in the baseline.  These elements are not attached to anything else so they do not affect the planarity of the construction.
\end{itemize}

Note that all of this can be axiomatized in first-order logic.

Now, we have to explain how $M'$ verifies that the value $N$ given in the input is indeed exactly the size of our structure. By counting all the elements above, we know that
our structure contains $2S + 4T + U$ elements in total, where $S = 3 \cdot 2^s$, 
$U$ is the number of 1s in $u$, and
$T = \Theta \left(|b_C| C + 2^s \log s + (s+|b_C|+|u|)^{O(1)} \right)$; 
the machine $M'$ can be designed so that it can compute its own running time $T$
precisely. 

Hence, if our formula has a model of size $N$, this $N$ must indeed be equal to 
the number encoded as $b_N$ in the initial tape, and this number must be accepted
by the machine $M$.

Now, let $N$ be a number such that the binary representation of $N$ is accepted 
by the machine $M$. If we assume that $|u| = s^2$ and $|b_C| = s$, the value of 
$N$ computed by $M'$ is of form
$\hat{N} = U + f(C, s) + g(s)$,
where $f(C, s) = \Theta(C \cdot s)$ and $g(s) = \Theta(2^s \log s)$.
Take $s$ such that $S = 3 \cdot 2^s = \Theta(N / \log N)$, and such that if we take
the maximum possible values of $C = 2^s-1$ and $U=s^2$, the obtained $\hat{N}$
is greater than $N$. 
By replacing these values by smaller ones ($C$ is obtained by division, and
$U$ is used for exact padding), we get $\hat{N} = N$; since $f(C,s)$ is the dominant
component of $\hat{N}$, our $C$ will still be $\Theta(N / \log(N))$. Since we have
assumed the computation of $M$ can be done in time and space $O(N / \log N)$,
these values of $C, S = \Theta(N / \log N)$ will be sufficient (if the machine $M$
actually needs less time for the given input, we can change it to non-deterministically
wait, so that there is always a run of the exact length we need).
Hence, if the binary representation of $N$ is
accepted  by the machine $M$, there is a model of our formula of size exactly $N$.
\end{proof}

\subsection{Queue machines}

In this section, we show that the class $\PPifSpec_2$ includes the
complexity class $\binnts(N)$, which is a result incomparable with
Theorem~\ref{clocksim}.  The result is established by considering
\emph{queue machines}, which we define next.

A \emph{queue machine} is a non-deterministic Turing machine with two
heads: a read head and a write head.  Both heads can only move to the
right.  The machine starts in an initial state with an input string on
the tape; the read head pointing to the first symbol of the string and
the write head at the first tape cell to the right of the string.  The
transitions are given as $4$-tuples $(q,w_1,w_2,q') \in Q \times
\Sigma^+ \times \Sigma^+ \times Q$, where $Q$ is the set of states of
the machine and $\Sigma$ is the tape alphabet.  A transition
$(q,w_1,w_2,q')$ is enabled if the machine is in state $q$ and the
string starting with the symbol under the read head begins with $w_1$.
If the transition is taken, then the read head moves right to the
symbol just after this occurrence of $w_1$, the string $w_2$ is added
to the write of the current string on the tape, the write head
moves to just past the end of this string, and the machine changes
state to $q'$.  The machine halts when it enters a final state.

It is clear that queue machines are Turing-universal.  However, we are
interested in the following measure of complexity: what is the total
number of symbols (including the initial tape contents) on the tape at
the end of the computation.  Note that no symbol is erased from the
tape, it is only the heads that move to the right.  As an example,
consider a machine with two state $q$ and $q_F$ of which the latter is
final, and the following transitions: $(q,b,a,q)$, $(q,a,ab,q)$,
$(q,A,Ab,q)$, $(q,A,AA,q_F)$.  Suppose this machine is started with
initial tape contents $A$.  One possible run of the machine ends with
the tape containing $AAbAbaAbaabAbaababaAA$.  It can be shown that in
any computation of this machine the number of symbols at the end of
the computation is a Fibonacci number, and conversely, there is such a
computation for each Fibonacci number.  To be precise, successive $A$s
(apart from the last one) appear exactly at positions $F-1$ where $F$
is a Fibonacci number.

For a queue machine $M$, let $S_M\subseteq \bbN$ denote the set of
numbers $n$ such that some computation of $M$ ends with exactly $n$
symbols on the tape.
Let $\QTL$ denote the collection of sets which are $S_M$ for some
queue machine $M$.  We can then establish the following inclusions.

\begin{thm}\label{queuemach}
$\binntime(\sqrt N) \subseteq \binnts(N) \subseteq \QTL \subseteq \PPifSpec_2$
\end{thm}
\begin{proof}
The first inclusion $\binntime(\sqrt N) \subseteq \binnts(N)$ follows
from definitions, since a Turing machine running in time $\sqrt N$
cannot take space more than $\sqrt N$ and thus the product of time and
space required is at most $N$.

For the second inclusion, we can use a simulation similar to that in
the proof of Theorem~\ref{clocksim}.  That is, from a machine $M$
accepting a set in $\binnts(N)$, we can first construct a machine $M'$
that clocks the number of steps and also clocks the space required,
while simulating $M$.  The machine $M'$ halts when the total space
(i.e.\ the length of the string on the tape, summed over all
configurations that have occurred in the computation) is exactly $N$,
the number that was represented in binary in the initial
configuration.   Now, it is easily checked that a queue machine
simulating $M'$ (in the standard way) will halt with exactly $N$
symbols on the tape if the binary representation of $N$ is accepted by
$M$. 

For the third inclusion, we show that, given a queue machine $M$, we
can axiomatise in first-order logic a structure $W$ with two PIFs that
represents the string that is on the tape at the end of a successful
computation of $M$.  One PIF is the successor function, connecting
position $x$ on the tape to $x+1$.  The other is a function $f_{rw}$
which maps a position which is the position of the read head at some
time $t$ to the position of the write head at the same time.  In
addition, we have a number of unary relations to code the content of
the tape and the state of the machine.  To see that the structure described is planar, consider the
embedding that places all points in a line, with the edges
corresponding to the successor function following the line.  The edges
corresponding to $f_{rw}$ connect to the source above the line and
the target below the line, looping around the left end of the line.
Since, moving from time $t$ to time $t+1$, both the read and the write
heads move to the right, it can be seen that such edges do not cross.
Axiomatising such a structure in first-order logic is now a
straightforward matter: we need to note that if a position $x$ is the
source of an $f_{rw}$ edge and $x$ encodes the fact that at this point
the machine is in state $q$, then for some transition $(q,w_1,w_2,q')$,
the string $w_1$ appears in the cells starting at $x$, the string
$w_2$ appears in the tape to the right of $f_{rw}(x)$, that $f_{rw}$
is defined at $x + |w_1|$ and nowhere else in between $x$ and  $x +
|w_1|$, and $f_{rw}(x+|w_1|) = f_{rw}(x)+|w_2|$.
\end{proof}

Note that the assumption of planarity is essential to the proof
above.  The first-order axioms we construct from a machine $M$ will
admit models that, in addition to $N$ elements representing a valid
computation of $M$ contain a loop containing a sequence of symbols
$wx$ such that $M$ when reading $wx$ outputs $xw$.  This loop will be necessarily
non-planar, though embeddable on a torus.  Indeed a model of the
axioms could admit any number of such additional loops.  On the other
hand, if $M$ is a machine which (by some measure) always increases the
length of the string at each step, then such a loop becomes
impossible.  The example given above generating sequences of Fibonacci
number length is one such machine.  Thus, one obtains this class as a
planar spectrum without having to assume planarity.  Indeed, 
Example~\ref{exfib} (in its planar form) is essentially encoding this
idea. 

\subsection{Upper bounds}

We now establish an upper bound that shows, in particular, that not every
spectrum is a planar spectrum.  Let $\PSpec$ denote the class of
\emph{planar spectra}, i.e.\ those sets $S \subseteq \mathbb{N}$ for
which there is a first-order sentence $\phi$ such that $N \in S$ if,
and only if, there is a planar model of $\phi$ with $N$ elements.  In
particular, $\PPifSpec_d \subseteq \PSpec$ and $\PBDSpec_d \subseteq
\PSpec$ for any $d$.

\begin{thm}\label{planarupper}
 $\PSpec \subseteq \binntime(N \log^2 N)$.
 \end{thm}
\begin{proof}
  It is known that, for any first-order sentence $\phi$ there is a
  linear-time algorithm that decides, for a planar graph $G$ whether
  or not $G \models\phi$~\cite{FG01}, for instance, on a random-access
  machine (RAM).  Since the average degree of a planar graph is at
  most $6$, a planar graph on $N$ vertices has at most $3N$ edges.
  Thus, a non-deterministic RAM can, on input $N$ guess a graph $G$
  with at most $3N$ edges, check that the graph is planar and then
  check determine whether $G \models \phi$, all in time $O(N)$.  For a
  linear time algorithm for checking that a graph is planar, we can
  use the algorithm from~\cite{FOP06}.  Since any set decidable in
  linear time on a non-deterministic RAM is in $\binntime(N \log^2
  N)$~\cite{Monien77}, the result follows.
\end{proof}

From the time hierarchy theorem~\cite{SFM78}, it follows that
$\binntime(N \log^2 N)$ is a proper subset of $\NE$ and since the
latter is exactly the class of all spectra, there are spectra which
are not planar spectra.  This answers negatively Open Question 6 from~\cite{fiftyyears}.

With respect to bounded-degree planar spectra, we note that the set of
such spectra with degree bounded by $3$ is already quite rich.

\begin{thm}\label{planardegree}
For any $d$, if $S \in \PPifSpec_d$ then $S' = \{2dn \mid n \in S\}
\in \PBDSpec_3$.
\end{thm}
 \begin{proof}
 This can be established by an argument similar to that for Corollary~\ref{simpcor}.  We cannot literally use the construction from the proof of Theorem~\ref{thsimplify}, as there is no guarantee that the structure $G$ obtained in the proof of that theorem is planar, even if $\stra$ is.  In that construction we build a gadget corresponding to each element $a$ of $\stra$ consisting of points $p_1,\ldots,p_d,q_1,\ldots,q_d$ connected in an $E$-path \emph{in that order}.  It is easily seen that we can, instead, choose a different order for each $a$ which will guarantee that $G$ is planar if $\stra$.  Indeed such an order is given by the order of incidence of edges on $a$ in a planar embedding of $\stra$.  Moreover, since we have unary relations $P_i$ identifying the vertices $p_i$ and $q_i$, the order is not relevant and we are still able to interpret $G$ in $\stra$ and vice versa.
 \end{proof}

\section{Forcing planarity}

In this section, we turn to considering the classes of spectra of the form $\FPPifSpec_d$ for some $d$.   We have previously seen examples that demonstrate that some sparse sets of numbers (such as the powers of 2 or the Fibonacci numbers) are realized as spectra of sentences $\phi$ such that all models of $\phi$ are planar.  However, in order to prove the general lower bounds established in Theorems~\ref{clocksim} and~\ref{queuemach}, we had to explicitly restrict attention to planar models as this is not something that was enforced by the axioms.  Indeed, in both cases, the axioms admit non-planar models where there are components disjoint from the machine simulation that can be embedded on a torus.   We now show that we are able to construct a machine simulation which can be axiomatised by a formula that enforces planarity, provided that we severely restrict the memory size: it needs to be linear in the size of the input (i.e.\ the binary representation of the number $N$).  We also need to restrict the time that the machine can run more stringently than in our previous results.

\begin{thm} \label{thm:forcing}
$\binntisp(N^{1-\epsilon}, \log N) \subseteq \FPPifSpec_2$ for all $\epsilon$ with $0< \epsilon < 1$.
\end{thm}

\begin{proof}
  Suppose we have a nondeterministic Turing machine $M$ which accepts the binary representation of $N$ in time $N^{1-\epsilon}$ and space $\log N$.  Let $d > 10$ be an integer such that $d > 2^{\frac{1}{\epsilon}}$.  We can assume, by changing alphabet as necessary, that machine $M$ uses space $\log_d N$.  Given a structure in a vocabulary with two PIFs that defines a $(N^{1-\epsilon}) \times \log_d N$ grid, we can, in a standard way, give first-order axioms saying that certain unary relations on the grid elements code an accepting computation of $M$.  Thus, to establish the result, it suffices to show that we can axiomatise, in first-order logic, a planar structure of size $N$ containing a definable grid of dimension $(N^{1-\epsilon}) \times \log_d N$.  This is what we do next.

We build on the basic spiral structure constructed in the proof of Theorem~\ref{makeorder}.  Indeed, the two PIFs are exactly those defined there, we will only add unary relations to the structure.  We assume that the elements of the structure are $A = \{1,\ldots,N\}$.  Let $n$ be the largest integer such that $2^n < N$ and we divide the set  $A$ into layers $L_1,\ldots,L_n,L_{n+1}$ where, for $k\leq n$, $L_k$ contains the elements in the range $[2^{k-1},2^k-1]$ and $L_{n+1}$ contains the elements in the range $[2^n,N]$.  We now inductively define two functions $v: A \rightarrow \{0,\ldots,4\}$ and $f: A \rightarrow \{0,\ldots,d-1\}$.

For $x \in L_{n+1}$, we let $v(x)=1$ and $f(x)=0$.

For $x\in L_n$, we let $v(x) = 2 + v(2x) + v(2x+1)$, where we assume $v(y)$ denotes $0$ if $y > N$.  Again, we let $f(x) = 0$.  Note that $\sum_{x\in L_n} v(x) = N+1$.

For $m < n$ and $x\in L_m$, we define $f(x) = \sum_{\{y \in L_{m+1} \mid y \geq 2x\}}v(y) \pmod d$.
Note that $f(x) = (v(2x) + v(2x+1) + f(x+1)) \mod d$. We then define $v(x)$ by
$v(x) = (v(2x) + v(2x+1) + f(x+1) - f(x)) / d \in \{0, 1\}$. In other words, we compute the sum of
numbers $v(y)$ for all $y$ in the layer $L_{m+1}$, from right to left, and store the suffix
sums as $f(x)$ in the layer $L_m$. Whenever the calculated suffix sum exceeds $d$, we
subtract $d$, and set $v(x)$ to 1. Hence, the sum of $v(x)$ for $x \in L_m$ will equal
the sum of $v(y)$ for $y \in L_{m+1}$ divided by $d$, and $f(2^{m-1})$ will be the
remainder.


Note that $f(1), f(2), \ldots, f(2^{n-1}), f(2^n)$ then contains the
representation of $N$ in base $d$.  Since $d>2$, there is $m_0$ such
that $f(2^k) = 0$ for all $k<m_0$.  Let $\Gamma = \{2^i x: x \in L_{m_0}, i \geq 0\}$.
For $x \in \Gamma$, let $g(x)$ be the smallest $y \in \Gamma$ such that $y>x$, but 
is still in the same layer -- if there is none, $g(x)$ is undefined.

We add $d$ unary relation symbols to our vocabulary interpreted by the $d$ sets $f(x)=0, f(x)=1, \ldots, f(x)=d-1$; five unary relations to represent 
$v(x)=0, \ldots, v(x)=4$, and one interpreted by the set $\Gamma$.  It can be checked that there are first-order axioms that ensure that the interpretation of these relation symbols is exactly that.

Now, the set $\Gamma$, together with the two PIFs $g$ and $f_{2x}$ is a grid of dimension $2^{m_0} \times (n-m_0)$. 
Moreover, $n-m_0$ is just the number of digits in the base $d$ representation of $N$, i.e.\ $\lceil \log_d N \rceil$.  Thus we have
$$2^{m_0} \sim 2^{log_2 N - log_d N} = 2^{log_2 N (1 - log_d 2)} = N ^ {1 - log_d 2} > N^{1-\epsilon},$$
where the last inequality follows from the fact that $d > 2^{\frac{1}{\epsilon}}$.  Thus, this gives us the grid we require.  However, $g$ is not definable in FO from $f_{+1}$ and $f_{2x}$.   Nonetheless, we can axiomatise the fact that a certain set $\calR$ of unary relations codes the computation of the machine $M$ on the grid defined by $g$ and $f_{2x}$ on $\Gamma$.   We can say that if the interpretation of $\calR$ at a position $x = 2^i t$ encodes the contents of tape cell $i$ at time $t$, then position $2^i t+1$ encodes it at the next time step.  This can then be propagated to the position $2^i (t+1)$  (i.e.\ $g(x)$) by means of the axiom
$\forall_{x \notin \Gamma} R(x+1) \iff R(x)$ for each $R \in \calR$.
\end{proof}

As an example, it follows that the set of prime numbers is the spectrum of a sentence using only two PIFs and all of whose models are planar graphs (of degree at most 4).  This is the case because the straightforward primality test of checking all potential divisors from $2$ to $\sqrt N$ shows that the set of primes is in $\binntisp(N^{\frac{1}{2}}, \log N)$.
\section{Conclusions and future work}
We have shown that, in considering the spectra of first-order
sentences, restricting the models to have bounded degree or to be
planar (or, indeed, both in combination) is a real restriction.  We
can put a complexity upper bound on these spectra that is strictly
lower than the $\NE$ characterization we know for the class of all
spectra.  At the same time, we have demonstrated that these classes of
spectra are still very rich by showing that we can encode certain
classes of machine simulations in them.  There are, however, gaps
between the lower bounds and the upper bounds we have established and
an obvious question for future work is whether either set of bounds
can be tightened.  It would be interesting to obtain an exact
characterisation in terms of complexity for the various classes of
spectra that we  consider.  Such tight characterisations seem unlikely
with Turing machines, but it might be possible to obtain them with
machine models with a more robust notion of linear time such as, for
example, Kolmogorov-Uspenski machines
(see~\cite{Gurevich-Kolmogorov}). 

If we consider spectra of sentences which only have planar models, we can still establish a lower bound, as in Theorem~\ref{thm:forcing}.  This is based on a simulation of machines with a severe space restriction, and it is quite possible that this bound is not tight.  Thus, a natural question is whether we can construct simulations of machines with higher memory requirements in such cases.

\section*{Acknowledgment}
The research reported here was carried out during a visit by the
second author to the University of Cambridge in 2011, supported by the ESF Resaearch Networking Programme GAMES. Eryk Kopczy{\'n}ski was also partially supported by the Polish National Science research grant DEC-2012/07/D/ST6/02435.

\bibliographystyle{alpha}
\bibliography{spectra}{}

\end{document}